\documentclass[a4paper,aps,pra,showpacs,twocolumn,superscriptaddress]{revtex4-1}

\usepackage[utf8]{inputenc}
\usepackage[T1]{fontenc}
\usepackage[british]{babel}
\usepackage{lmodern}
\usepackage[scaled=1.03]{inconsolata}
\usepackage{xcolor}
\definecolor{darkred}{rgb}{0.5,0,0}
\definecolor{darkgreen}{rgb}{0,0.5,0}
\definecolor{darkblue}{rgb}{0,0,0.5}
\usepackage[colorlinks,breaklinks,linkcolor=darkred,citecolor=darkgreen,urlcolor=darkblue]{hyperref}
\usepackage{graphicx}
\usepackage{tikz}
\usepackage[babel]{microtype}
\usepackage{amsmath,amssymb,amsthm,bm,mathtools,amsfonts,mathrsfs,bbm,dsfont}
\usepackage{xspace}
\usepackage{multirow}
\usepackage{verbatim}

\usepackage{physics}
\newcommand{\id}{\ensuremath{\mathds{1}}}
\usepackage{bbold}

\renewcommand{\vec}[1]{\boldsymbol{#1}}
\renewcommand{\leq}{\leqslant}

\renewcommand{\geq}{\geqslant}
\renewcommand{\ge}{\geqslant}
\newcommand\hb[2]{\genfrac{}{}{0pt}{}{#1}{#2}}
\newcommand{\mi}{\mathrm{i}}
\newcommand{\me}{\mathrm{e}}
\newcommand{\md}{\mathrm{d}}

\newtheoremstyle{mystyle}{6pt}{6pt}{\normalfont}{0pt}{\bf}{.}{ }{}

\theoremstyle{mystyle}

\newtheorem{observation}{Observation}

\newtheorem{example}{Example}

\begin{document}

\nonfrenchspacing

\title{Quantum measurement incompatibility in subspaces}

\author{Roope Uola}
\affiliation{D\'{e}partement de Physique Appliqu\'{e}e, Universit\'{e}  de Gen\`{e}ve, CH-1211 Gen\`{e}ve, Switzerland}

\author{Tristan Kraft}
\affiliation{Naturwissenschaftlich-Technische Fakult\"at, Universit\"at Siegen, Walter-Flex-Str.~3, D-57068 Siegen, Germany}

\author{S\'ebastien Designolle}
\affiliation{D\'{e}partement de Physique Appliqu\'{e}e, Universit\'{e}  de Gen\`{e}ve, CH-1211 Gen\`{e}ve, Switzerland}

\author{Nikolai Miklin}
\affiliation{International Centre for Theory of Quantum Technologies (ICTQT), University of Gdansk, 80-308 Gda\'nsk, Poland}

\author{Armin~Tavakoli}
\affiliation{D\'{e}partement de Physique Appliqu\'{e}e, Universit\'{e}  de Gen\`{e}ve, CH-1211 Gen\`{e}ve, Switzerland}

\author{Juha-Pekka Pellonp\"a\"a}
\affiliation{QTF Centre of Excellence, Turku Centre for Quantum Physics, Department of Physics and Astronomy,
  University of Turku, FI-20014 Turun yliopisto, Finland
}

\author{Otfried Gühne}
\affiliation{Naturwissenschaftlich-Technische Fakult\"at, Universit\"at Siegen, Walter-Flex-Str.~3, D-57068 Siegen, Germany}

\author{Nicolas Brunner}
\affiliation{D\'{e}partement de Physique Appliqu\'{e}e, Universit\'{e}  de Gen\`{e}ve, CH-1211 Gen\`{e}ve, Switzerland}

\date{\today}

\begin{abstract}
We consider the question of characterising the incompatibility of sets of high-dimensional quantum measurements.
We introduce the concept of measurement incompatibility in subspaces.
That is, starting from a set of measurements that is incompatible, one considers the set of measurements obtained by projection onto any strict subspace of fixed dimension.
We identify three possible forms of incompatibility in subspaces: (i) \textit{incompressible incompatibility}: measurements that become compatible in every subspace, (ii) \textit{fully compressible incompatibility}: measurements that remain incompatible in every subspace, and (iii) \textit{partly compressible incompatibility}: measurements that are compatible in some subspace and incompatible in another.
For each class we discuss explicit examples.
Finally, we present some applications of these ideas.
First we show that joint measurability and coexistence are two inequivalent notions of incompatibility in the simplest case of qubit systems.
Second we highlight the implications of our results for tests of quantum steering.
\end{abstract}

\maketitle

\section{Introduction}

Quantum theory is built on Hilbert spaces, in which observables are presented as Hermitian operators and states as positive unit-trace matrices.
This gives the theory a noncommuting structure, resulting in, for example, various uncertainty relations and different notions of measurement incompatibility.
In the simplest case of Hermitian operators, all incompatibility is captured by the concept of noncommutativity, but for more general measurements given by positive operator-valued measures (or POVMs for short), various possibilities arise.
These include measurement disturbance~\cite{HW10}, joint measurability~\cite{Bus86,Lah03,HMZ16} and coexistence~\cite{Lud54,LP97}.
Of these, joint measurability is probably the most well known.
Loosely speaking, a set of POVMs is called jointly measurable when there exists a single parent POVM, from which one can recover the statistics of all POVMs in the set.
This concept has found many applications in quantum information theory, notably through connections to quantum nonlocality~\cite{WPF09,BV18,HQB18}, quantum steering~\cite{UMG14,QVB14,UBGP15}, macrorealism~\cite{UVB19}, and temporal and channel steering~\cite{KPUR15,BDM15,ULGP18}, as well as in prepare-and-measure scenarios~\cite{CHT18,CHT19,SSC19,GQA19,UKS+19,CHT20,TU20}.

Recently, the notion of joint measurability has been investigated for measurements on high-dimensional systems~\cite{CHT12,Haa15,BQG+17,BN18,DSFB19,CHT19,DFK19,NDBG20,Kiu20,DSU+20}, which allow in principle for stronger incompatibility compared to the case of qubits.
This raises the question of whether one could define a notion of ``dimensionality'' for measurement incompatibility.
In particular, given a set of incompatible POVMs, can the incompatibility be localised in specific subspaces of lower-dimensional POVMs or is it, on the contrary, an intrinsic property of the high-dimensional space?
To formalise this problem we introduce the idea of measurement incompatibility in subspaces.
That is, given a set of non jointly measurable POVMs, we project (i.e., truncate) each POVM onto a lower dimensional subspace and investigate the compatibility properties of the resulting set of projected POVMs.
We identify all possible forms of measurement incompatibility under this scenario, which can be of three types: (i) \textit{incompressible incompatibility}, i.e., measurements that become compatible in every strict subspace, (ii) \textit{fully compressible incompatibility}, i.e., measurements that remain incompatible in every nontrivial subspace, and (iii) \textit{partly compressible incompatibility}, i.e., measurements that are compatible in some subspace and incompatible in another.
We present explicit examples of all three categories of incompatibility in subspaces.

Beyond the fundamental interest, we show that these ideas have applications.
First, taking advantage of an example of partly compressible incompatibility, we show that the notions of joint measurability and coexistence (i.e., joint measurability of \textit{all} binarisations of the involved measurements) are inequivalent in the simplest case of qubit POVMs.
This answers a long-standing open question on the relation between these notions~\cite{Lud54,LP97}.
For binary or extremal measurements the concepts are known to coincide, even when using one extremal continuous variable measurement in the latter case~\cite{CHT05}.
On the contrary, for general measurements in qutrit systems and beyond, the concepts are known to be inequivalent~\cite{RRW13}.
We solve the missing qubit scenario.
It is worth noting that there was no reason to expect this result since other incompatibility notions, such as noncommutativity and unavoidable measurement disturbance, are known to be inequivalent only from dimension three on~\cite{HW10}.

Second, these ideas have an impact on quantum correlations, in particular the notion of quantum steering~\cite{WJD07,CS16b,UCNG20}, which is directly connected to measurement incompatibility~\cite{QVB14,UMG14,UBGP15}.
We discuss the role of dimension in the context of this connection.
The latter states that a party performing an incompatible set of measurements can always steer another party via a well-chosen bipartite quantum state.
We point out that the connection cannot be directly applied to scenarios where the steered party has a system of lower dimension compared to that of the steering party.

\section{Preliminaries}

To introduce measurement incompatibility, we first fix the notation.
A measurement assemblage $\mathcal M=\{M_{a_x|x}\}_{a_x,x}$ consists of POVMs, i.e., Hermitian positive semidefinite matrices, such that for every $x$ one has $\sum_{a_x} M_{a_x|x}=\id$, acting on a finite-dimensional Hilbert space.
Here $\id$ is the identity operator, $x$ labels the choice of measurement, and $a_x$ is the corresponding outcome.
POVMs give rise to measurement statistics in a given quantum state $\varrho$ through the formula $p(a_x|x,\varrho)=\tr(M_{a_x|x}\varrho)$.
When there is no risk of confusion, we substitute $a_x$ with $a$.

This formalism motivates the definition of joint measurability of a measurement assemblage $\mathcal M$ as the possibility of obtaining the statistics of any measurement in $\mathcal M$ from a common parent measurement~\cite{HMZ16}.
Any outcome of the parent measurement is a list $\vec a$ of outcomes of single measurements, and the statistics of a single measurement is obtained by summing over certain parts of the list.
Formally, joint measurability of $\mathcal M$ is defined as the existence of a parent POVM $\{G_{\vec a}\}_{\vec a}$ such that
\begin{align}
  M_{a_x|x}=\sum_{\hb{a_i}{i\neq x}} G_{\vec a}.
\end{align}
Measurements that do not allow a parent POVM of this form are called not jointly measurable or incompatible.

The concept of joint measurability is best illustrated with an example.
In a qubit system, let us take a measurement assemblage corresponding to the noisy versions of the binary spin measurements $\sigma_x$ and $\sigma_z$, i.e., $M^\mu_{\pm|1}=\frac{1}{2}(\id\pm\mu\sigma_x)$ and $M^\mu_{\pm|2}=\frac{1}{2}(\id\pm\mu\sigma_z)$.
Here the parameter $\mu\in[0,1]$ quantifies the amount of noise.
For these measurements, a natural candidate of a parent POVM is given by a procedure, where a 50-50 beam splitter (or a coin) decides between the measurement directions $x+z$ and $x-z$~\cite{ULMH16}.
The resulting statistics are described by the POVM
\begin{align}
  G_{i,j}=\frac{1}{4}\Big[\id+\frac{1}{\sqrt{2}}(i\sigma_x+j\sigma_z)\Big],
\end{align}
where $i,j\in\{-1,1\}$.
It is straightforward to verify that ignoring the outcome $j=\pm1$ results in $M^{1/\sqrt 2}_{\pm|1}$ and similarly ignoring the outcome $i=\pm1$ results in $M^{1/\sqrt 2}_{\pm|2}$.
Hence, one has a joint measurement for the noisy spin measurements with $\mu=1/\sqrt2$.
It can be shown that this threshold is indeed optimal in the sense that there is no parent POVM when $\mu>1/\sqrt{2}$~\cite{Bus86}.

\section{Measurement incompatibility in subspaces}

We are interested in the following problem: given an incompatible measurement assemblage $\mathcal M=\{M_{a_x|x}\}_{a_x|x}$ acting in a $d$-dimensional Hilbert space with $2<d<\infty$, what happens to the incompatibility when the assemblage is truncated to an $n$-dimensional subspace with $2\leq n<d$?
The truncation is modelled by a projection $P_n$ onto an $n$-dimensional subspace, i.e., we are interested in the compatibility properties of the measurement assemblage
\begin{equation}
    \mathcal M_n^U=\{P_nU^\dagger M_{a_x|x}UP_n\}_{a_x|x},
\end{equation}
where $U$ is some unitary operator acting on the initial $d$-dimensional Hilbert space and $P_n=\sum_{i=1}^n|i\rangle\langle i|$.
Note that in contrast to truncating quantum states for which the normalisation, i.e., unit-trace property, can be altered, for POVMs the normalisation is unaltered when seen as measurements in the subspace, i.e., the normalisation is the identity operator in the subspace.

We find that there exist three different forms of incompatibility in subspaces.
First, compatibility can be present in \textit{all} strict subspaces of dimension $n$, i.e., $\mathcal M_n^U$ being compatible for every unitary $U$ in the initial Hilbert space.
Second, incompatibility can be present in \textit{all} strict subspaces of dimension $n$, i.e., $\mathcal M_n^U$ being incompatible for every unitary $U$.
Finally, there is the possibility of having compatibility for some unitary $U$ and incompatibility for some other unitary $V$.
Note that compatible measurement assemblages fulfil the first notion trivially, i.e., a parent measurement $G_{\vec a}$ of $\mathcal M$ becomes a parent measurement $P_n U^\dagger G_{\vec a}UP_n$ for the truncated assemblage $\mathcal M_n^U$, see also \cite{Kiu20}.
To clarify the different types of incompatibility in subspaces, we discuss each category in detail below.

\subsection{Incompressible incompatibility}

We first show the existence of sets of measurements that become compatible in any strict subspace.
Hence incompatibility is incompressible here, as it vanishes in every possible lower-dimensional subspace.
Intuitively, this represents the most fragile form of incompatibility in subspaces.

Formally, we are searching for an incompatible measurement assemblage $\mathcal M$, with the property that the truncation $\mathcal M_n^U$ is compatible for every $U$ and $n<d$.

Here we provide a method for constructing such assemblages for the case $d=3$ (and, hence, $n=2$).
To this end, we use the connection between measurement incompatibility and quantum steering.
More specifically, we start from the steering scenario and consider the so-called stronger Peres conjecture~\cite{Pus13}.
The latter was recently disproven~\cite{MGHG14} (see also~\cite{VB14,YO17}), and we make use of these results to construct a measurement assemblage that is incompressible.

The stronger Peres conjecture states that every bound entangled state admits a local hidden state model~\cite{WJD07}, i.e., cannot lead to quantum steering.
In other words, given a bound entangled state $\varrho_{AB}$, i.e., an entangled state that cannot be distilled into a pure entangled state, together with any measurement assemblage $\{A_{a_x|x}\}_{a_x,x}$, one is conjectured to have
\begin{equation}\label{eqn:LHS}
  \sigma_{a_x|x}:=\tr_A[(A_{a_x|x}\otimes\id)\varrho_{AB}]=\sum_{\hb{a_i}{i\neq x}}\sigma_{\vec{a}},
\end{equation}
where $\sigma_{\vec{a}}$ are positive operators with the property $\sum_{\vec{a}}\sigma_{\vec{a}}=\tr_A(\varrho_{AB})=:\varrho_B$.
The operators $\sigma_{\vec{a}}$ are referred to as the local hidden states, and with the marginalisation in Eq.~\eqref{eqn:LHS}, they form a local hidden state model for the state assemblage $\sigma_{a_x|x}$.
The latter is then called unsteerable.
If no such model can be constructed, the assemblage is steerable.
It has turned out that the existence of a local hidden state model is equivalent to the joint measurability of the corresponding ``pretty good measurements'' $M_{a_x|x}:=\varrho_B^{-1/2}\sigma_{a_x|x}\varrho_B^{-1/2}$~\cite{UBGP15}.
On the contrary, if $\sigma_{a_x|x}$ is steerable, then the pretty good measurement is incompatible.

With these tools we are ready to explain our construction.
We start from the counterexample to the stronger Peres conjecture in the two-qutrit case presented in Ref.~\cite{UBGP15}.
This features a specific bound entangled state $\varrho_{AB}$ which, combined with well-chosen measurements, leads to a steerable assemblage.
The corresponding pretty good measurements are therefore incompatible.
This set of POVMs is in fact incompressible.
That is, a projection onto any possible qubit subspace will necessarily give a jointly measurable set of POVMs.
To see this, we note that the state assemblage corresponding to the truncated measurements can be obtained from the original steering setup by projecting the steered side of the bound entangled state to a qubit subspace.
As $\varrho_{AB}$ is positive under partial transposition~\cite{Per96}, the state resulting from the local projection is necessarily a separable state~\cite{HHH96}.
As separable states can only lead to unsteerable assemblages, it follows that the corresponding projected pretty good measurements are jointly measurable, which concludes the proof.
Note that we provide a detailed proof in Appendix~\hyperref[app:peres]{A}.

It is worth mentioning that one can modify the concept of incompressible incompatibility by demanding that a measurement assemblage becomes compatible under every Heisenberg channel (a channel preserving identity, but not necessarily trace preserving) to a smaller dimensional system.
Although we will leave open the question of whether this provides a strict subset of measurement assemblages that are compatible in every subspace, we note that the Peres conjecture technique also works in this scenario, see Appendix~\hyperref[app:peres]{A}.
The channel formulation of incompatibility in subspaces turns out to be relevant when applying the concept to quantum steering.

Intuitively, incompressible incompatibility can be viewed as a weak form of incompatibility.
This can be formalised more quantitatively by considering a measure of incompatibility, the so-called depolarising incompatibility robustness~\cite{DFK19}.
This measure corresponds to the critical amount of depolarising noise one needs to add to incompatible measurements to make them compatible, namely,
\begin{align}\label{eqn:primal}
 \eta^{\mathrm{d}}_{\{A_{a|x}\}}=\max_{\eta,\{G_{\vec{j}}\}_{\vec{j}}} &\quad \eta \nonumber \\
  \text{s.t.~} &\quad \sum_{\vec{j}} \delta_{j_x,a} G_{\vec{j}} = A_{a|x}^\eta \quad \forall a,x, \\
  &\quad G_{\vec{j}} \geq 0 \quad \forall {\vec{j}}, \quad \eta \leq 1, \nonumber
\end{align}
where $A_{a|x}^\eta=\eta A_{a|x}+(1-\eta)\tr(A_{a|x})\id/d$.
The measure is clearly nonnegative and it equals one for compatible measurements; the lower it is, the more incompatible the measurements are.
The intuition of the approach below is to average over all lower-dimensional parent POVMs in order to get one for the initial measurements; naturally in the process some noise appears so that the resulting parent measurement actually gives a lower bound on the depolarising incompatibility robustness.

We consider a measurement assemblage $\{A_{a|x}\}_{a,x}$ such that for all projections $P_n$ onto an $n$-dimensional subspace ($n>1$) of the $d$-dimensional space in which the measurements live, there exists a parent POVM $G_{\vec{j}}^{(P_n)}$ for the measurement assemblage $\{P_nA_{a|x}P_n\}_{a,x}$.
Then we have that
\begin{equation}\label{eqn:parent}
  G_{\vec{j}}:=\frac{d}{n}\int G_{\vec{j}}^{(P_n)}\md P_n
\end{equation}
is a parent POVM for the measurements with elements
\begin{equation}
  \eta_nA_{a|x}+(1-\eta_n)\frac{\tr(A_{a|x})}{d}\id,\quad\text{with}\quad\eta_n=\frac{nd-1}{d^2-1}
\end{equation}
so that the depolarising incompatibility robustness admits a lower bound
\begin{equation}
    \eta_{\{A_{a|x}\}}^\mathrm{d}\geq\frac{nd-1}{d^2-1}.
\end{equation}
Note that for $n=d$, we indeed get the expected trivial bound of one.
We give the proof of the above bound in the case $n=2$, as the general case of $n>1$ can be obtained through an iterative procedure.
We decompose any projection $P_2$ into $\ketbra{\varphi}+\ketbra{\psi}$ where $\ket{\varphi}$ and $\ket{\psi}$ are orthogonal.
Note that the integration $\int\md P_2$ used above should always be thought of as $\iint\md\psi\md\varphi$, where $\ket{\psi}$ lives in the $(d-1)$-dimensional subspace orthogonal to $\ket{\varphi}$.
Note also that the integral notation for operators is a convenient formal tool that nonetheless needs some caution: it always underpins the complex integrals obtained by sandwiching it with two vectors.
Then the marginals of the proposed parent POVM~\eqref{eqn:parent} are
\begin{align}
  &\sum_{\vec{j}}\delta_{j_x,a}G_{\vec{j}}\nonumber\\
  =\,&\frac{d}{2}\int\sum_{\vec{j}}\delta_{j_x,a}G_{\vec{j}}^{(P_2)}\md P_2\text{~~~~by linearity}\nonumber\\
  =\,&\frac{d}{2}\int P_2A_{a|x}P_2\md P_2\text{~~~~by assumption}\nonumber\\
  =\,&\frac{d}{2}\iint \big(\ketbra{\varphi}+\ketbra{\psi}\big)A_{a|x}\big(\ketbra{\varphi}+\ketbra{\psi}\big)\md\psi\md\varphi.\nonumber\\
\end{align}
Since $\ket{\psi}$ lives in the subspace orthogonal to $\ket{\varphi}$ we have that $\int\ketbra{\psi}\md\psi=(\id-\ketbra{\varphi})/(d-1)$.
Therefore we get
\begin{equation}\label{eqn:margint}
  \sum_{\vec{j}}\delta_{j_x,a}G_{\vec{j}}=\frac{A_{a|x}}{d-1}+\frac{d(d-2)}{d-1}I(A_{a|x}),
\end{equation}
where
\begin{equation}\label{eqn:I}
    I(M):=\int\ketbra{\varphi}M\ketbra{\varphi}\md\varphi.
\end{equation}

Computing the integral~\eqref{eqn:I} requires some care.
Consider a Hermitian operator written in its diagonal basis
\begin{equation}
    M=\sum_{i=1}^d\lambda_i\ketbra{i}.
\end{equation}
In the following, we will make use of the book~\cite{Rud80} by Rudin on function theory on the complex unit ball; since his vocabulary is quite different from ours, we establish the connection in detail.
We start by explaining how $\bra{i}[I(\ketbra{i})]\ket{i}$ can be computed, with the rest being similar.
So we aim at evaluating
\begin{equation}
    \bra{i}\big[I(\ketbra{i})\big]\ket{i}=\int|\braket{\varphi}{i}|^4\md\varphi,
\end{equation}
which writes, in the language of Ref.~\cite{Rud80},
\begin{equation}
    \int|\zeta^\alpha|^2\md\sigma(\zeta),
\end{equation}
where the variable $\zeta=(\braket{\varphi}{i})$ and the multi-index $\alpha=(2)$ contain only one element in this case.
Then, Proposition~1.4.9(1) from~\cite{Rud80} guarantees that, as $|\alpha|=2$,
\begin{equation}
    \bra{i}\big[I(\ketbra{i})\big]\ket{i}=\frac{(d-1)!\alpha!}{(d-1+|\alpha|)!}=\frac{2}{d(d+1)}.
\end{equation}
For $j\neq i$, the same argument applies with $z=(\braket{\varphi}{i},\braket{\varphi}{j})$ and $\alpha=(1,1)$ now containing two elements so that
\begin{equation}
    \bra{j}\big[I(\ketbra{i})\big]\ket{j}=\frac{1}{d(d+1)}.
\end{equation}
For the off-diagonal elements, Proposition~1.4.8 from Ref.~\cite{Rud80} indicates that they are zero.
Combining things together we get
\begin{align}
    I(M)&=\sum_{i=1}^d\lambda_iI(\ketbra{i})\\
    &=\sum_{i=1}^d\lambda_i\frac{\ketbra{i}+\sum_j\ketbra{j}}{d(d+1)}\\
    &=\frac{M+\tr(M)\id}{d(d+1)},
\end{align}
so that, by plugging this expression in Eq.~\eqref{eqn:margint} we eventually get
\begin{equation}
  \sum_{\vec{j}}\delta_{j_x,a}G_{\vec{j}}=\frac{2d-1}{d^2-1}A_{a|x}+\frac{d(d-2)}{d^2-1}\tr(A_{a|x})\frac{\id}{d},
\end{equation}
which concludes the proof.

In the generalisation to projections with a higher rank, the following integrals are obtained:
\begin{equation}
    \frac{d}{n}\int P_n\md P_n=\id,
\end{equation}
\begin{equation}
    \frac{d}{n}\int P_n M P_n\md P_n=\frac{(nd-1)M+(d-n)\tr(M)\id}{d^2-1},
\end{equation}
\begin{equation}
    \frac{d}{n}\int\tr(P_nMP_n)P_n\md P_n=\frac{(d-n)M+(nd-1)\tr(M)\id}{d^2-1}.
\end{equation}
Note that $\md P_n$ is an abusive notation that should be understood as $\md\varphi_n\ldots\md\varphi_1$, where each $\ket{\varphi_k}$ lives in the $(d-k+1)$-dimensional subspace orthogonal to all $\ket{\varphi_j}$ with $j<k$.

\subsection{Fully compressible incompatibility}

Let us now discuss a completely different form of incompatibility in subspaces, namely, sets of measurements that remain incompatible in every lower-dimensional subspace.
Intuitively, this represents the most robust form of incompatibility in subspaces.

Formally, we are searching for an incompatible measurement assemblage $\mathcal M$, with the property that the truncation $\mathcal M_n^U$ is incompatible for every $U$ and $n\geq2$.
We present a sufficient criterion for measurements to be of this type when truncated from dimension $d$ to dimension $d-1$.
The criterion works for measurements constructed from orthonormal bases of the Hilbert space, i.e., measurements for which every POVM element is of the form $|\varphi_a\rangle\langle\varphi_a|$ for some basis $\{|\varphi_a\rangle\}_{a=1}^d$.

To derive our criterion, let $\{|\varphi_n\rangle\}_{n=1}^d$ (with $d\geq3$) be an arbitrary orthonormal basis of $\mathcal H$ and $\{|\varphi'_k\rangle\}_{k=1}^d$ another orthonormal basis such that 
\begin{equation}\label{eqn:ehto1}
  \langle\varphi_n|\varphi'_k\rangle\ne 0 \qquad \text{for all $n$ and $k$}.
\end{equation}
Define two $d$-valued (rank-one) projection valued measures (PVMs)  $\{P_n\}_n$ and $\{P'_k\}_k$ with $P_n=|\varphi_n\rangle\langle\varphi_n|$ and $P'_k=|\varphi'_k\rangle\langle\varphi'_k|$.
They are totally noncommutative: 
\begin{equation}\label{eqn:kaava1}
  P_nP'_k=\underbrace{\langle\varphi_n|\varphi'_k\rangle}_{\ne 0}|\varphi_n\rangle\langle\varphi'_k|
  \ne \underbrace{\langle\varphi'_k|\varphi_n\rangle}_{\ne 0}|\varphi'_k\rangle\langle\varphi_n|=
  P'_kP_n
\end{equation}
for all $n$ and $k$ (since the ranges of $P_nP'_k$ and $P'_kP_n$ are disjoint: $\mathds{C}|\varphi_n\rangle\cap \mathds{C}|\varphi'_k\rangle=\{0\}$).
Hence, $P$ and $P'$ are not jointly measurable.

Let $|\psi\rangle$ be an arbitrary unit vector in the $d$-dimensional space and define the projection $R=\id-|\psi\rangle\langle\psi|$ onto the (arbitrary) $(d-1)$-dimensional closed subspace $R\mathcal H$ (i.e., $\mathcal H=R\mathcal H\oplus\mathds{C}|\psi\rangle$).
We have three cases.

First, $|\psi\rangle\langle\psi|$ commutes with all $P_n$'s, i.e., $|\psi\rangle\langle\psi|=|\varphi_m\rangle\langle\varphi_m|$ for some $m$ (since the rank-one projection $|\psi\rangle\langle\psi|$ must be diagonal in the basis $\{|\varphi_n\rangle\}_{n=1}^d$).
Now also $R$ commutes with $P$ so that
the projections $RP_n R$ constitute a $(d-1)$-valued rank-one PVM of $R\mathcal H$ (since $RP_m R=0$).
Moreover, $\{RP'_k R\}_k$ is a $d$-valued rank-one POVM of $R\mathcal H$ (note that 
  $R|\varphi'_k\rangle=\sum_{n\ne m} \langle\varphi_n|\varphi'_k\rangle
|\varphi_n\rangle\ne 0$ for all $k$).
Similarly as in Eq.~\eqref{eqn:kaava1}, one sees that $RP_nR\,RP'_kR\ne RP'_kR\,RP_nR$ for all $n\ne m$ and for all $k$ so that the projected observables are not jointly measurable  (recall that a PVM and a POVM are jointly measurable if and only if they commute).

Second, $|\psi\rangle\langle\psi|$ commutes with all $P'_k$'s, i.e., $|\psi\rangle\langle\psi|=|\varphi'_\ell\rangle\langle\varphi'_\ell|$ for some $\ell$.
Exactly as in the preceding item (just change the roles of the bases $|\varphi_n\rangle\longleftrightarrow |\varphi'_k\rangle$ ) one sees that the (projected) PVM and POVM do not commute and thus are not jointly measurable.

Third, suppose that $|\psi\rangle\langle\psi|$ is not $|\varphi_n\rangle\langle\varphi_n|$ or $|\varphi'_k\rangle\langle\varphi'_k|$ for any $n$ or $k$, that is, $R$ does not commute with $P$ or $P'$.
Now $RP_n R\ne 0$ and $RP'_k R\ne 0$ for all $n$ and $k$ (indeed, suppose that, e.g., $0=R|\varphi_n\rangle=|\varphi_n\rangle-\langle\psi|\varphi_n\rangle|\psi\rangle$, i.e., $|\psi\rangle=c|\varphi_n\rangle$, with $c\in\mathds{C}$ such that $|c|=1$, i.e., $|\psi\rangle\langle\psi|=|\varphi_n\rangle\langle\varphi_n|$, a contradiction).
Hence, both $\{RP_n R\}_n$ and $\{RP'_k R\}_k$ are $d$-valued rank-one POVMs (not PVMs) of $R\mathcal H$ with the minimal Naimark dilations $(\mathcal H,P,R)$ and $(\mathcal H,P',R)$.
Assume that they have a joint POVM $\{M_{nk}\}_{n,k}$, i.e.,
$\sum_{k=1}^dM_{nk}=RP_n R$ and $\sum_{n=1}^dM_{nk}=RP'_k R$.
From Ref.~\cite{Pel14} one sees that there are unique numbers $a_{nk}\ge 0$ and $b_{nk}\ge 0$ such that $\sum_{k=1}^d a_{nk}=1$ for all $n$ and $\sum_{n=1}^d b_{nk}=1$ for all $k$ and
\begin{equation}
  M_{nk}=a_{nk}RP_n R=b_{nk}RP'_k R.
\end{equation}
From $RP_n R\ne 0\ne RP'_k R\ne 0$ one gets $a_{nk}=0$ if and only if $b_{nk}=0$.

Since $\sum_{k=1}^d a_{nk}=1$ (for all $n$) we must have $a_{n\pi(n)}\ne 0$ for some index $k=\pi(n)$.
Hence, for each $n=1,\ldots,d$, $0\ne R|\varphi_n\rangle=c_n R|\varphi'_{\pi(n)}\rangle$, i.e., $R(|\varphi_n\rangle-c_n|\varphi'_{\pi(n)}\rangle)=0$, for some complex number $c_n\ne 0$.
Hence, $|\varphi_n\rangle-c_n|\varphi'_{\pi(n)}\rangle$ belongs to the kernel of $R$ and is of the form $c'_n|\psi\rangle$, $c'_n\ne 0$, by condition \eqref{eqn:ehto1}, so that
\begin{equation}
  |\psi\rangle=a_n|\varphi_n\rangle+b_n|\varphi'_{\pi(n)}\rangle\qquad\text{for all $n$,}
\end{equation}
where $a_n$ and $b_n$ are some \emph{nonzero} (by the assumption) complex numbers such that $\|\psi\|=1$.
The constants $a_n$ and $b_n$ are unique since $|\varphi_n\rangle$ and $|\varphi'_{\pi(n)}\rangle$ are linearly independent by condition \eqref{eqn:ehto1}.
It is easy to show that $\pi$ is a permutation (bijection)  on $\{1,2,\ldots,d\}$ (such that $a_{n\pi(n)}\ne 0$ for all $n$); indeed, if $\pi(n)=\pi(m)$ then $(b_n-b_m)|\varphi'_{\pi(n)}\rangle=a_m|\varphi_m\rangle-a_n|\varphi_n\rangle$ which forces $b_n=b_m$ and then $a_m=0=a_n$ yielding a contradiction: $|\psi\rangle=b_n|\varphi'_{\pi(n)}\rangle$.

Now $|\psi\rangle\in\bigcap_{n=1}^d(\mathds{C}|\varphi_n\rangle+\mathds{C}|\varphi'_{\pi(n)}\rangle)$ which we want to be $\{0\}$ \emph{for all permutations} $\pi$ (a contradiction since $\psi\ne 0$).
Taking $\langle\varphi_m|\psi\rangle$, we have
\begin{equation}
  \quad\quad a_n\delta_{nm}+b_n\langle\varphi_m|\varphi'_{\pi(n)}\rangle=a_j\delta_{jm}+b_j\langle\varphi_m|\varphi'_{\pi(j)}\rangle
\end{equation}
for all $n,\,m,\,j$.
In particular, if $n\ne m\ne j$, we have $b_n\langle\varphi_m|\varphi'_{\pi(n)}\rangle=b_j\langle\varphi_m|\varphi'_{\pi(j)}\rangle$ or 
\begin{equation}
  \frac{b_n}{b_j}=\frac{\langle\varphi_m|\varphi'_{\pi(j)}\rangle}{\langle\varphi_m|\varphi'_{\pi(n)}\rangle}
\end{equation}
where the left hand side does not depend on $m$.
If we choose any $n\ne j\ne k\ne n$ (which is possible since $d\ge 3$) and write $J=\pi(j)$, $K=\pi(k)$, $N=\pi(n)$ (so that $N\ne J\ne K\ne N$ since $\pi$ is bijective)
we get
\begin{equation}
  \frac{\langle\varphi_m|\varphi'_J\rangle}{\langle\varphi_m|\varphi'_N\rangle}=\frac{b_n}{b_j}=\frac{b_n}{b_k}\frac{b_k}{b_j}=\frac{\langle\varphi_o|\varphi'_K\rangle}{\langle\varphi_o|\varphi'_N\rangle}\frac{\langle\varphi_p|\varphi'_J\rangle}{\langle\varphi_p|\varphi'_K\rangle}
\end{equation}
or
\begin{equation}
  \langle\varphi_m|\varphi'_J\rangle\langle\varphi_o|\varphi'_N\rangle\langle\varphi_p|\varphi'_K\rangle=\langle\varphi_p|\varphi'_J\rangle\langle\varphi_m|\varphi'_N\rangle\langle\varphi_o|\varphi'_K\rangle
\end{equation}
for all $j\ne m\ne n\ne o\ne k\ne p\ne j$ (which is possible since $d\ge 3$).
To conclude, we have to find an orthonormal basis satisfying \eqref{eqn:ehto1} and the following (sufficient) condition: for all  $N\ne J\ne K\ne N$ and $m \ne o \ne p \ne m$
\begin{equation}\label{eqn:ehto2_app}
  \langle\varphi_m|\varphi'_J\rangle\langle\varphi_o|\varphi'_N\rangle\langle\varphi_p|\varphi'_K\rangle\ne \langle\varphi_p|\varphi'_J\rangle\langle\varphi_m|\varphi'_N\rangle\langle\varphi_o|\varphi'_K\rangle.
\end{equation}
In other words if conditions \eqref{eqn:ehto1} and \eqref{eqn:ehto2_app} are satisfied then $P$ and $P'$ are incompatible PVMs in a $d$-dimensional Hilbert space, with all projections onto $(d-1)$-dimensional subspaces also incompatible.

\begin{example}
  Let $d=3$ and $\{|\varphi_n\rangle\}_{n=1}^3$ be the computational basis of $\mathds{C}^3$.
  Now
  \begin{eqnarray*}
    |\varphi'_1\rangle &=& \frac{1}{\sqrt{14}}(1 , 2  , 3 ) \\
    |\varphi'_2\rangle &=& \frac{1}{\sqrt{27}}(-5 ,1 ,1 ) \\
    |\varphi'_3\rangle &=& \frac{1}{\sqrt{378}}(1 ,16 ,-11 )
  \end{eqnarray*}
  clearly satisfy conditions \eqref{eqn:ehto1} and \eqref{eqn:ehto2_app}.
\end{example}

Intuitively, measurements featuring fully compressible incompatibility should be very incompatible.
The question of quantifying measurement incompatibility has been recently formalised~\cite{DFK19}, and it appears that pairs of measurements based on two mutually unbiased bases (MUBs) are among the most incompatible ones.
Surprisingly, the following example shows that a pair of MUBs is not fully compressible (but only partly compressible).
This shows that incompatibility in subspaces captures a different aspect of measurement incompatibility compared to the usual quantifiers.

\begin{example}
  Continuing with the notation from the previous example, the Fourier connected vectors
  \begin{equation}
    |\varphi'_k\rangle=\frac{1}{\sqrt d}\sum_{n=1}^d \me^{\frac{2\mi\pi nk}{d}}|\varphi_n\rangle
  \end{equation}
  do not satisfy \eqref{eqn:ehto2_app}.
  Now, if $d=3$ and we choose $\ket{\psi}=\frac{1}{\sqrt{3}}(1,1,\omega)$, where $\omega=\exp(2\mi\pi/3)$, we get
  \begin{equation}
    R=\id_3-\ket{\psi}\bra{\psi}=\frac19
    \begin{pmatrix}
      2 &-1 &-\omega^2 \\
      -1& 2&-\omega^2 \\
      -\omega& -\omega& 2\\
    \end{pmatrix}
  \end{equation}
  and
  \begin{equation}
    R\ket{\varphi_1}\bra{\varphi_1} R =\frac19
    \begin{pmatrix}
      4 &-2 &-2\omega^2 \\
      -2& 1&\omega^2 \\
      -2\omega& t& 1\\
    \end{pmatrix}
    =R\ket{\varphi'_1}\bra{\varphi'_1} R
  \end{equation}
  and similarly
  \begin{equation}
    R\ket{\varphi_n}\bra{\varphi_n}R=R\ket{\varphi'_n}\bra{\varphi'_n} R\qquad \text{for all $n=1,2,3$,}
  \end{equation}
  i.e., the projected POVMs $\{R\ket{\varphi_n}\bra{\varphi_n}R\}_n$ and $\{R\ket{\varphi'_n}\bra{\varphi'_n}R\}_n$ are the same POVM, which makes them trivially jointly measurable.
  As a technical note, the truncated POVM has two minimal Naimark dilations $(\mathcal H,\ket{\varphi_n}\bra{\varphi_n},R)$ and $(\mathcal H,\ket{\varphi'_n}\bra{\varphi'_n},R)$, the only difference being the projective measurement in the dilation space, namely the fact that one can measure either $\{\ket{\varphi_n}\bra{\varphi_n}\}_n$ or $\{\ket{\varphi'_n}\bra{\varphi'_n}\}_n$ in any subsystem's state $\varrho$ to get the same statistics.
\end{example}

\subsection{Partly compressible incompatibility}

Together with the two extreme scenarios, it is possible to have incompatible measurements in dimension $d$, which become compatible or incompatible depending on the truncation.
Arguably, this represents the least surprising (and probably most common) form of incompatibility in subspaces.

Formally, we are searching for a measurement assemblage $\mathcal M$ that is incompatible, with the property that the truncation $\mathcal M_n^U$, with fixed $2\leq n<d$, is compatible for some $U$ and incompatible for some other $\tilde U$.
The most naive way of finding such examples is to add together (as a direct sum) a compatible and an incompatible measurement assemblage.
Now, projections onto the components of the direct sum yield measurement assemblages that have different compatibility properties, i.e., one is compatible and one is incompatible.
We note that this structure can be also realised through the concept of commutativity domain, as characterised by a theorem of Ylinen \cite{Yli85}.

There are however less trivial examples, and we will discuss one of them in detail when demonstrating the inequivalence between joint measurability and coexistence for qubit POVMs.

\section{Inequivalence of joint measurability and coexistence for qubits}

The notion of coexistence of measurement assemblages goes back to Ludwig~\cite{Lud54}.
Ludwig's original formulation is in the language of measure theory, which we will omit here to avoid technicalities.
Instead, we use the fact that the concept can be recast as joint measurability of all yes-no questions (or binarisations) of a given measurement assemblage~\cite{HMZ16}.
Recall that a binarisation of a POVM $\{M_{a}\}_a$ with respect to an outcome subset $X$ is a two-valued POVM $\{\sum_{a\in X}M_{a},\sum_{a\notin X}M_{a}\}$.
For jointly measurable assemblages the parent POVM gives an answer to all questions before binarisation, so clearly joint measurability implies coexistence.
The problem of identifying scenarios in which these two notions do not coincide has formed its own research program.
In certain cases, including projective, binary, and extremal measurements, these two notions coincide \cite{Lah03,Pel14,HPU15}.
Up to now, two classes of examples of coexistent measurement assemblages that are incompatible have been reported~\cite{RRW13,Pel14}.
These classes work for systems whose dimension is three or larger.
Here, we extend one of these classes to the missing qubit case.
Our solution goes as follows.
Take two POVMs on a qutrit system that are known to be incompatible and coexistent~\cite{Pel14}:
\begin{align}
  A_i:=&\frac{1}{2}(\id - |i\rangle\langle i|),\ i=0,1,2\label{eqn:CoexA}\\
  B_j:=&\begin{cases}\frac{1}{2}|j\rangle\langle j|,\ &j=0,1,2\\
  \frac{1}{2}|\psi_{j-3}\rangle\langle\psi_{j-3}|,\ &j=3,4,5,\label{eqn:CoexB}\end{cases}
\end{align}
where $\{|i\rangle\}_{i=0}^2$ is the computational basis and $|\psi_j\rangle=\frac{1}{\sqrt 3}(|0\rangle+\omega^j|1\rangle+\omega^{2j}|2\rangle)$ with $\omega=\exp(2\mi\pi/3)$ is its Fourier-connected basis.
These measurements are coexistent, as every binarisation of the measurement given by Eq.~\eqref{eqn:CoexA} gives an element ($j=0,1,2$) of the measurement in Eq.~\eqref{eqn:CoexB}.
This shows that $B$ functions as a parent measurement for all binarisations of both measurements.
More precisely, we note that joint measurability can be equivalently formalised as the existence of a POVM $\{G_\lambda\}_\lambda$ and classical post-processings, i.e., probability distributions, $p(a|x,\lambda)$ such that
\begin{align}\label{eqn:JMMarkov}
  M_{a|x}=\sum_\lambda p(a|x,\lambda)G_\lambda.
\end{align}
To see that this definition is equivalent to our main definition, one can define a joint measurement from the r.h.s.~of Eq.~\eqref{eqn:JMMarkov} through $G_{\vec a}=\sum_\lambda \Pi_x p(a|x,\lambda)G_\lambda$.
As any POVM is a joint measurement of its own binarisations in the sense of Eq.~\eqref{eqn:JMMarkov}, and as the binarisations of the measurements given by Eq.~\eqref{eqn:CoexA} and Eq.~\eqref{eqn:CoexB} are binarisations of the latter, we have proven their coexistence.
As the incompatibility of these measurements is proven in~\cite{Pel14} and can also be deduced from the subsequent discussion, we omit the proof here.

To find the desired qubit example, we analyse the compatibility of these measurements in the two-dimensional subspace spanned by $\ket{\psi_0}$ and $\ket{\psi_1}$.
Under the projection  $P_2=|\psi_0\rangle\langle\psi_0|+|\psi_1\rangle\langle\psi_1|$
to this subspace, the POVMs transform as follows
\begin{align}
  A_i\mapsto&\quad\frac{1}{6}|\psi_0+\bar{\omega}^{i+1}\psi_1\rangle\langle\psi_0+\bar{\omega}^{i+1}\psi_1|\\
  +&\quad\frac{1}{6}|\psi_0+\bar{\omega}^{i+2}\psi_1\rangle\langle\psi_0+\bar{\omega}^{i+2}\psi_1|\nonumber\\
  B_j\mapsto&\begin{cases}\frac{1}{6}(|\psi_0+\bar{\omega}^{j}\psi_1\rangle\langle\psi_0+\bar{\omega}^{j}\psi_1|),\ &j=0,1,2\\ P_2 B_j P_2=B_j,\ &j=3,4\\ 0,\ &j=5.\end{cases}
\end{align}
We let $\tilde A$ and $\tilde B$ be the matrix representations of the restrictions of the POVMs $A$ and $B$ in the basis $\{|\psi_0\rangle,|\psi_1\rangle\}$.
Indeed,
\begin{align}\label{eqn:coexJMBloch}
    \tilde A_0&=\frac{1}{2}\Big(\frac{2}{3}\openone-\frac{1}{3}\sigma_x \Big)\\
    \tilde A_1&=\frac{1}{2}\Big(\frac{2}{3}\openone+ \frac16 \sigma_x + \frac1{2\sqrt3} \sigma_y \Big)\\
    \tilde A_2&=\frac{1}{2}\Big(\frac{2}{3}\openone+ \frac16 \sigma_x - \frac1{2\sqrt3} \sigma_y \Big)\\
    \tilde B_0&=\frac{1}{2}\Big(\frac13 \openone+\frac13\sigma_x\Big)\\
    \tilde B_1&=\frac{1}{2}\Big(\frac13 \openone -\frac16 \sigma_x -\frac1{2\sqrt3}\sigma_y\Big)\\
    \tilde B_2&=\frac{1}{2}\Big(\frac13 \openone -\frac16 \sigma_x +\frac1{2\sqrt3}\sigma_y\Big)\\
    \tilde B_3&=\frac{1}{2}\Big(\frac12\openone+\frac12 \sigma_z\Big)\\
    \tilde B_4&=\frac{1}{2}\Big(\frac12\openone-\frac12 \sigma_z\Big).
\end{align}
We note that $\tilde A$ and $\tilde B$ are coexistent due to the fact that they are projections of a coexistent measurement assemblage.
Formally, the former parent of the binarisations is truncated to a parent of the projected assemblages.
To prove incompatibility of the truncated measurement assemblage, we note that the operators $\{\tilde B_0, \tilde B_1, \tilde B_2, \tilde B_3\}$ are linearly independent and one can write $\tilde B_4$ as their linear combination,
\begin{align}\label{eqn:lindep}
  \tilde B_4=\tilde B_0 + \tilde B_1 + \tilde B_2 - \tilde B_3.
\end{align}
Using Eq.~\eqref{eqn:lindep} and assuming that $\tilde A$ and $\tilde B$ are jointly measurable, we get (as $\tilde B$ is rank-one)~\cite{Pel14}
\begin{align}
  \tilde A_i=\sum_{j=0}^4 p_{ij}\tilde B_j &= (p_{i0}+p_{i4})\tilde B_0 + (p_{i1}+p_{i4})\tilde B_1\nonumber\\
  &\,+ (p_{i2}+p_{i4})\tilde B_2 + (p_{i3}-p_{i4})\tilde B_3.
\end{align}
As the operators on the r.h.s.~are linearly independent, we get for $\tilde A_0$ the coefficients $p_{00}=p_{03}=p_{04}=0$ and $p_{01}=p_{02}=1$ and for $\tilde A_1$ the coefficients $p_{11}=p_{13}=p_{14}=0$ and $p_{10}=p_{12}=1$.
This is already a contradiction as $p_{02}+p_{12}=2 > 1$.
Hence the truncated measurements are coexistent, but not jointly measurable.

Note that a final coarse graining can be applied without losing this feature, namely, one can group the first two outcomes of $\tilde B$ so as to get a four-valued measurement.
The coexistence is obviously preserved in the process, and the incompatibility can be shown by computing the depolarising incompatibility robustness, see Eq.~\eqref{eqn:primal}, which is approximately 0.9830.
Hence, we have constructed a counterexample for the coexistence problem in the qubit case including one three outcome and one four outcome POVM.
However, one might wonder whether there exists a smaller counterexample, i.e., one with less outcomes.
We have not been able to find any projection preserving the incompatibility of the example from Ref.~\cite{RRW13}, which features two and three outcomes.
We have also explored this question numerically via a seesaw method consisting of two semidefinite programs (SDP), but haven't found such examples.

Below we explain the seesaw algorithm for a pair of POVMs $\{A_i\}_i^{m_a}$ and $\{B_j\}_j^{m_b}$, but it can be extended to the case of three or more measurements.
The algorithm starts by sampling two random POVMs, which we denote as $\{A^{(0)}_i\}_i^{m_a}$ $\{B^{(0)}_j\}_j^{m_b}$.
For this pair of POVMs we construct an incompatibility witness as follows:
\begin{align}
  \max_{X_i,Y_j,N}\quad & \sum_i^{m_a}\tr(X_i A^{(0)}_i)+\sum_j^{m_b}\tr(Y_j B^{(0)}_j)\label{eqn:seesaw_1}\\
  \text{s.t.}\quad & X_i\geq 0,\; \forall i\in [m_a],\nonumber\\ & Y_j\geq 0,\; \forall j\in [m_b],\nonumber\\
  & X_i + Y_j \leq N,\; \forall i\in[m_a], j\in[m_b],\nonumber\\
  & \tr N=1,\; N^\dagger = N.\nonumber
\end{align}
This is the dual SDP to the generalised incompatibility robustness~\cite{UKS+19}.

Let us denote the solutions to the above SDP as $\{X^{(1)}_i\}_i^{m_a}$ and $\{Y^{(1)}_j\}_j^{m_b}$.
Now, the second SDP in our seesaw algorithm is designed to look for coexistent POVMs which would maximise the witness $\{X^{(1)}_i\}_i^{m_a}$ and $\{Y^{(1)}_j\}_j^{m_b}$.
This SDP reads as follows:
\begin{align}
  \max_{G_\lambda}\quad & \sum_i^{m_a}\tr(X^{(1)}_i A_i)+\sum_j^{m_b}\tr(Y^{(1)}_j B_j)\label{eqn:seesaw_2}\\
  \text{s.t.}\quad & G_\lambda \geq 0,\; \forall \lambda,\nonumber\\
  & \sum_\lambda D(S^a|\lambda)G_\lambda = \sum_{i\in S^a} A_i,\; \forall S^a \subset [m_a],\nonumber\\
  & \sum_\lambda D(S^b|\lambda)G_\lambda = \sum_{j\in S^b} B_j,\; \forall S^b \subset [m_b],\nonumber\\
  & \sum_\lambda G_\lambda = \id,\; \sum_i^{m_a}A_i = \id,\; \sum_j^{m_b}B_j = \id.\nonumber
\end{align}
In the above SDP the coexistence of POVMs $\{A_i\}_i^{m_a}$ $\{B_j\}_j^{m_b}$ is ensured by joint measurably of every binarisation of the latter.
With a slight abuse of notations, the post-processing functions $D(S^a|\lambda)$ should satisfy $D(S^a|\lambda)+D([m_a]\setminus S^a|\lambda) = 1, \forall \lambda$, and $D(S^b|\lambda)+D([m_b]\setminus S^b|\lambda) = 1, \forall \lambda$.
As usual, these post-processing functions can be taken to be deterministic.
The POVMs $\{A_i\}_i^{m_a}$ $\{B_j\}_j^{m_b}$ are auxiliary variables of the SDP in Eq.~\eqref{eqn:seesaw_2} since they are defined as linear functions of $G_\lambda$.
However, we are interested in these POVMs which come from the solutions of the SDP given by Eq.~\eqref{eqn:seesaw_2}.
Let us denote these solutions as $\{A^{(1)}_i\}_i^{m_a}$ $\{B^{(1)}_j\}_j^{m_b}$.
The final step of defining the seesaw algorithm is the imputing $\{A^{(1)}_i\}_i^{m_a}$ $\{B^{(1)}_j\}_j^{m_b}$ to the SDP in Eq.~\eqref{eqn:seesaw_1} and iterating the process until the value of the objective function converges to some point.
If at any point the solution of the SDP in Eq.~\eqref{eqn:seesaw_1} returns a value larger than $1$, an example of incompatible coexistent POVMs $\{A_i\}_i^{m_a}$ and $\{B_j\}_j^{m_b}$ is found.

With this method, we were able to find numerical examples for various dimensions of Hilbert space as well as various configurations.
For instance, for $d=3$ and the simplest case of one binary and one trinary POVM, the algorithm converges to examples for about $1\%$ of the initial random samples of the POVMs $\{A^{(0)}_i\}_i^{m_a}$ and $\{B^{(0)}_j\}_j^{m_b}$.
For a higher number of outcomes, the algorithm was more likely to find examples.
However, in the qubit case the algorithm could find only a weakly incompatible example for two POVMs with three and four outcomes.
Due to the low value of incompatibility we could not give an analytical form of this example.

Finally, note that a similar seesaw algorithm has been previously used to find examples of quantum states with interesting entanglement properties~\cite{MMG16}.

Note also that the incompatibility of the above example is only partly compressible since there exist projections onto qubit subspaces such that the resulting pair of measurements is compatible.
For instance, under the projection $P_2=\ketbra{0}+\ketbra{1}$, the incompatibility is lost.

\section{Implications for quantum steering}

Shifting our focus to the connection between joint measurability and steering, we pose a question on the role of dimension in this particular result.
Namely, it is known that any incompatible measurement assemblage on one party leads to a steerable state assemblage on the other party given that one possesses a suitable catalyst state.
The known catalyst states have full Schmidt rank.
Hence, one can raise the question of what happens if the dimension of the steered party is bounded.
Our examples of incompatible measurements that are compatible in every subspace allow one to answer this question.
Namely, one can see the shared state in a steering experiment as a Heisenberg channel (i.e., completely positive identity-preserving map) that maps one party's measurements to the pretty good measurements on the other party (up to transposition), i.e., $A_{a_x|x}\mapsto M_{a_x|x}^T$ according to~\cite{KBUP17}
\begin{align}
  \Lambda_{\varrho_{AB}}(A_{a_x|x})=\varrho_B^{-\frac12}\tr_A[(A_{a_x|x}\otimes\id)\varrho_{AB}]^T\varrho_B^{-\frac12},
\end{align}
where the transpose is taken in the eigenbasis of $\varrho_B=\tr_A(\varrho_{AB})$.
This channel is the Choi channel of $\varrho_{AB}$.
Clearly, any incompatible measurement assemblage that becomes compatible under any channel to a smaller-dimensional system does not enable steering when the steered party's dimension is smaller than the other party's dimension.

\section{Conclusions}

We have developed the concept of measurement incompatibility in subspaces.
We showed that this question leads to a rich structure, as truncated measurements can feature very different compatibility properties.

In particular we have shown the existence of sets of POVMs that have incompressible incompatibility, i.e., they become jointly measurable in every possible subspace.
We provided an example for this in dimension $d=3$, with projections for every qubit subspace ($n=2$).
It would be interesting to find other examples and see if this is possible in general, that is, for every $d$ and $n<d$.
Here the higher-dimensional counterexamples to the Peres conjecture of Ref.~\cite{YO17} might prove useful.

Another direction would be to characterise the sets of POVMs featuring different forms of incompatibility in subspaces.
The set of sets of POVMs with incompressible incompatibility should be convex.
What about others? Can one formalise witnesses for detecting different forms of incompatibility in subspaces?

It would also be interesting to see if incompatibility in subspaces is connected to the idea of compression with respect to a set of measurements~\cite{BRW18} or to genuine high-dimensional steering~\cite{DSU+20}.

Finally, we discussed some applications of these ideas.
First, we used an example of partly compressible incompatibility to show the inequivalence of joint measurability and coexistence in the simplest qubit case.
We also discussed the consequences for steering tests.
We conclude by noting that there are also other types of correlations that are closely related to various forms of measurement incompatibility such as preparation contextuality, Bell nonlocality, violations of macrorealism, and channel steering.
We believe that our framework can lead to a better understanding of these concepts and their applications.

\section{Acknowledgements}

We are thankful for Marco T\'ulio Quintino for useful comments on the manuscript.
We acknowledge the financial support from the Swiss National Science Foundation (Starting Grant DIAQ and NCCR SwissMap), the Finnish Cultural Foundation, the ERC (Consolidator Grant No.~683107/TempoQ),  
the Deutsche Forschungsgemeinschaft (DFG, German Research Foundation - Projects No.~447948357 and No.~440958198), and the Sino-German Center for Research Promotion, as well as the partial support by the Foundation for Polish Science (IRAP project, ICTQT, Contract No.~2018/MAB/5, co-financed by the EU within the Smart Growth Operational Programme).

\section{Note added}

While completing this work, we became aware of the recent and independent work in Ref.~\cite{LN20}.

\appendix

\section{The stronger Peres conjecture}
\label{app:peres}

In this Appendix we show how to construct incompatible measurements that are compatible in every two-dimensional subspace.
In Ref.~\cite{MGHG14} the authors have proven the existence of bound entangled steerable quantum states, hence providing a counterexample to the stronger Peres conjecture~\cite{Pus13}.
These states are given as
\begin{equation}\label{eqn:Perescounterstate}
\begin{aligned}
  \varrho_{AB}=&\lambda_1|\psi_1\rangle\langle\psi_1|+\lambda_2|\psi_2\rangle\langle\psi_2|\\
  &+\lambda_3(|\psi_3\rangle\langle\psi_3|+|\tilde\psi_3\rangle\langle\tilde\psi_3|),
\end{aligned}
\end{equation}
where
\begin{equation}
\begin{aligned}
  |\psi_1\rangle&=(|12\rangle+|21\rangle)/\sqrt2\\
  |\psi_2\rangle&=(|00\rangle+|11\rangle-|22\rangle)/\sqrt3\\
  |\psi_3\rangle&=m_1|01\rangle+m_2|10\rangle+m_3(|11\rangle+|22\rangle)\\
  |\tilde\psi_3\rangle&=m_1|02\rangle-m_2|20\rangle+m_3(|21\rangle-|12\rangle)
\end{aligned}
\end{equation}
and $m_i\geq0$.
As noted in Ref.~\cite{MGHG14}, this family of states can be made invariant under partial transposition on Alice's side by choosing
\begin{equation}\label{eqn:PerescounterPPTparameters}
\begin{aligned}
  \lambda_1&=1-\frac{2+3m_1m_2}{4-2m_1^2+m_1m_2-2m_2^2}\\
  \lambda_3&=\frac{1}{4-2m_1^2+m_1m_2-2m_2^2}\\
  \lambda_2&=1-\lambda_1-2\lambda_3.
\end{aligned}
\end{equation}
For positivity of the states, one has to check the limits on $m_i$.

The states in Eq.~\eqref{eqn:Perescounterstate} are steerable (at least for a certain range of parameters) with two measurements on Alice's side given by two MUBs~\cite{MGHG14}.
\begin{align}\label{eqn:Perescountermeasurements}
  |\varphi_{1,2|1}\rangle&=(1/\sqrt3,-1/\sqrt6,\pm1/\sqrt2)\nonumber\\
  |\varphi_{3|1}\rangle&=(1/\sqrt3,\sqrt{2/3},0)\nonumber\\
  |\varphi_{1|2}\rangle&=(1,0,0)\nonumber\\
  |\varphi_{2|2}\rangle&=(0,\omega/\sqrt2,\mi\omega/\sqrt2)\nonumber\\
  |\varphi_{3|2}\rangle&=(0,\overline\omega/\sqrt2,-\mi\overline\omega/\sqrt2),
\end{align}
where $\omega=\exp(2\mi\pi/3)$.

For our purposes, the steerability of the state assemblage
\begin{align}\label{eqn:Peresconjuctureassemblage}
  \varrho_{a|x}:=\tr_A[(|\varphi_{a|x}\rangle\langle\varphi_{a|x}|\otimes\id)\varrho_{AB}]
\end{align}
together with the partial transpose invariance of the state in Eq.~\eqref{eqn:Perescounterstate} are crucial.
Namely, if one maps Bob's side of the shared state into any qubit subspace, one is left with a separable state and, consequently, an unsteerable assemblage.

To take this idea a bit further, recall that steerability is very closely related to joint measurability of POVMs.
The connection is given by renormalisation of state assemblages, i.e., mapping $\varrho_{a|x}$ into $B_{a|x}:=\varrho_B^{-1/2}\varrho_{a|x}\varrho_B^{-1/2}$, where $\varrho_B=\tr_A(\varrho_{AB})$.
Note that here the state $\varrho_B$ is possibly inverted only on a subspace and, hence, the resulting POVMs $B_{a|x}$ are in general defined on a system of dimension less than or equal to Bob's original dimension.

Whereas the state assemblage in Eq.~\eqref{eqn:Peresconjuctureassemblage} originates from the state $\varrho_{AB}$, the normalised state assemblage (or POVMs) $\{B_{a|x}\}_{a|x}$ originates, up to a constant, from the filtered state $(\id\otimes\varrho_B^{-1/2})\varrho_{AB}(\id\otimes\varrho_B^{-1/2})/N$, where $N$ is the dimension of the support of $\varrho_B$.
As the original state $\varrho_{AB}$ is invariant under partial transposition on Alice's side, so is the filtered state.
Putting the known connection between steerability and joint measurability together with the fact that the filtered state is PPT and that PPT states in $\mathds{C}^3\otimes\mathds{C}^2$ are separable, we arrive at the following observation.

\begin{observation}
There exists an incompatible measurement assemblage in a qutrit system that becomes compatible under any restriction (i.e., CPTP mapping) to a qubit system.
\end{observation}

\begin{proof}
  Filter the state from Eq.~\eqref{eqn:Perescounterstate} with $\varrho_B^{-1/2}$ on Bob's side.
  Choosing the parameters as in Eq.~\eqref{eqn:PerescounterPPTparameters} results in a PPT state (because the state is invariant under partial transposition on Alice's side).
  Performing the measurements from Eq.~\eqref{eqn:Perescountermeasurements} on Alice's side leads to a filtered version of a steerable assemblage.
  This essentially corresponds (i.e., modulo possible normalisation constant due to the filter) to the pretty good measurements associated to the original state assemblage, which are incompatible.
  Hence, the filtered assemblage is steerable.

  Mapping this state assemblage into any two-dimensional subspace gives an assemblage, which originates from the filtered state together with a local map on Bob's side.
  As the resulting state is invariant under partial transposition on Alice's side, one gets a PPT state in $\mathds{C}^3\otimes\mathds{C}^2$, which is separable and consequently can only lead to unsteerable assemblages.
  Hence, the restricted assemblage is unsteerable for any CPTP map acting on Bob's side.

  To see the connection to joint measurability, notice that Bob's side of the filtered state is maximally mixed and, hence, the pretty good measurement link between joint measurability and steering corresponds to multiplication with a constant.
  To be more precise, take the assemblage from Eq.~\eqref{eqn:Peresconjuctureassemblage} and write
  \begin{align}
    B_{a|x}:=\varrho_B^{-1/2}\varrho_{a|x}\varrho_B^{-1/2}=\tr_A[(|\varphi_{a|x}\rangle\langle\varphi_{a|x}|\otimes\id)\varrho_{AB}^{\text{filt}}],
  \end{align}
  where $\varrho_{AB}^{\text{filt}}=(\id\otimes\varrho_B^{-1/2})\varrho_{AB}(\id\otimes\varrho_B^{-1/2})$.
  These observables are by definition not jointly measurable.
  Mapping these observables into any two-dimensional subspace gives
  \begin{align}
    \Lambda^\dagger(B_{a|x})=\tr_A[(|\varphi_{a|x}\rangle\langle\varphi_{a|x}|\otimes\id)(\id\otimes\Lambda^\dagger)(\varrho_{AB}^{\text{filt}})].
  \end{align}
  Note that the positive operator $(\id\otimes\Lambda^\dagger)\varrho_{AB}^{\text{filt}}$ is not normalised.
  However, the trace of this operator is equal to two.
  Putting this together with the PPT invariance, we see that the state assemblage $\tilde\varrho_{a|x}:=\frac{1}{2}\Lambda^\dagger(B_{a|x})$ is unsteerable.
  The steering equivalent observables of this assemblage are simply $\Lambda^\dagger(B_{a|x})$ as Bob's side of the state $(\id\otimes\Lambda^\dagger)\varrho_{AB}^{\text{filt}}/2$ is $\frac12\id_2$, where $\id_2$ is the identity operator in $\mathds{C}^2$.
\end{proof}

\end{document}